\pdfoutput=1
\documentclass[a4paper,UKenglish,thm-restate,pdfa]{lipics/lipics-v2021-modified}
\usepackage{mathtools}
\usepackage{tikz}
\usetikzlibrary{positioning}

\def\P{\ensuremath{\mathrm{P}}}
\def\NP{\ensuremath{\mathrm{NP}}}

\def\EE{\ensuremath{\mathrm{EE}}}
\def\NE{\ensuremath{\mathrm{NE}}}
\def\NEE{\ensuremath{\mathrm{NEE}}}
\def\FP{\ensuremath{\mathrm{FP}}}
\def\UP{\ensuremath{\mathrm{UP}}}
\def\DisjNP{\ensuremath{\mathrm{DisjNP}}}
\def\DisjCoNP{\ensuremath{\mathrm{DisjCoNP}}}

\def\coNP{\ensuremath{\mathrm{coNP}}}
\def\coNE{\ensuremath{\mathrm{coNE}}}
\def\coNEE{\ensuremath{\mathrm{coNEE}}}

\def\NPcoNP{\ensuremath{\mathrm{NP}\cap\mathrm{coNP}}}
\def\TALLY{\ensuremath{\mathrm{TALLY}}}
\def\NPMV{\ensuremath{\mathrm{NPMV}}}

\def\NPSV{\ensuremath{\mathrm{NPSV}}}

\def\NPbV{\ensuremath{\mathrm{NPbV}}}

\def\NPkV{\ensuremath{\mathrm{NP}k\mathrm{V}}}

\def\TAUT{\ensuremath{\mathrm{TAUT}}}
\def\SAT{\ensuremath{\mathrm{SAT}}}
\def\QBF{\ensuremath{\mathrm{QBF}}}
\def\PF{\ensuremath{\mathrm{PF}}}

\def\TFNP{\ensuremath{\mathrm{TFNP}}}

\def\PSPACE{\ensuremath{\mathrm{PSPACE}}}

\def\rsa{\ensuremath{\mathsf{RSA}}}
\def\B{\ensuremath{\Phi}}

\def\pls{\ensuremath{\mathrm{PLS}}}
\def\ppa{\ensuremath{\mathrm{PPA}}}
\def\ppp{\ensuremath{\mathrm{PPP}}}
\def\ppad{\ensuremath{\mathrm{PPAD}}}
\def\ppads{\ensuremath{\mathrm{PPADS}}}

\def\hUP{\ensuremath{\mathsf{UP}}}
\def\hDisjNP{\ensuremath{\mathsf{DisjNP}}}
\def\hDisjCoNP{\ensuremath{\mathsf{DisjCoNP}}}
\def\hNPcoNP{\ensuremath{\mathsf{NP}{}\cap{}\mathsf{coNP}}}
\def\hCON{\ensuremath{\mathsf{CON}}}
\def\hCONN{\ensuremath{\mathsf{CON}^{\mathsf{N}}}}
\def\hSAT{\ensuremath{\mathsf{SAT}}}
\def\hTFNP{\ensuremath{\mathsf{TFNP}}}
\def\hPneqNP{\ensuremath{\mathsf{P}\neq \mathsf{NP}}}
\def\hNPneqcoNP{\ensuremath{\mathsf{NP}\neq \mathsf{coNP}}}
\def\NPcoNP{\ensuremath{\mathsf{NP} \cap \mathsf{coNP}}}

\def\N{\ensuremath{\mathrm{\mathbb{N}}}}

\def\leqmpp{\ensuremath{\leq_\mathrm{m}^\mathrm{pp}}}
\def\leqmp{\ensuremath{\leq_\mathrm{m}^\mathrm{p}}}
\def\sim{\ensuremath{\leq_\mathrm{s}}}
\def\psim{\ensuremath{\leq_\mathrm{s}^\mathrm{p}}}
\def\leqlex{\ensuremath{\leq_\text{lex}}}
\def\lelex{\ensuremath{<_\text{lex}}}

\newcommand{\tn}[1]{\ensuremath{\text{#1}}}

\DeclareMathOperator{\dom}{dom}
\DeclareMathOperator{\ran}{ran}

\def\sqsubsetneq{\mathrel{\sqsubseteq\kern-0.92em\raise-0.15em\hbox{\rotatebox{313}{\scalebox{1.1}[0.75]{\(\shortmid\)}}}\scalebox{0.3}[1]{\ }}}
\def\sqsupsetneq{\mathrel{\sqsupseteq\kern-0.92em\raise-0.15em\hbox{\rotatebox{313}{\scalebox{1.1}[0.75]{\(\shortmid\)}}}\scalebox{0.3}[1]{\ }}}
\newcommand{\card}[1]{|#1|}
\DeclareMathOperator{\enc}{\mathrm{enc}}

\def\upprod{\texttt{handle\_UP\_stage}}
\def\zprod{\texttt{handle\_UP\_stage\_aftercare}}
\def\constprod{\texttt{construct}}

\newcommand{\xTrueThenElse}[2]{\if\x1{#1}\else{#2}\fi}

\let\origthelstnumber\thelstnumber 
\makeatletter
\newcommand*\Suppressnumber{%
  \lst@AddToHook{OnNewLine}{%
    \let\thelstnumber\relax%
     \advance\c@lstnumber-\@ne\relax%
    }%
}
\newcommand*\Reactivatenumber{%
  \lst@AddToHook{OnNewLine}{%
   \let\thelstnumber\origthelstnumber%
   \advance\c@lstnumber\@ne\relax}%
}
\makeatother

\makeatletter
\def\namedlabel#1#2{\begingroup
   \def\@currentlabel{#2}%
   \label{#1}\endgroup
}
\makeatother

\hideLIPIcs

\def\thetitle{An Oracle with no $\mathrm{UP}$-Complete Sets, but $\mathrm{NP}=\mathrm{PSPACE}$}
\title{\thetitle}
\titlerunning{\thetitle}

\author{David Dingel}{Julius-Maximilians-Universität Würzburg, Germany}{david.dingel@uni-wuerzburg.de}{https://orcid.org/0009-0009-6684-682X}{}
\author{Fabian Egidy}{Julius-Maximilians-Universität Würzburg, Germany}{fabian.egidy@uni-wuerzburg.de}{https://orcid.org/0000-0001-8370-9717}{supported by the German Academic Scholarship Foundation}
\author{Christian Glaßer}{Julius-Maximilians-Universität Würzburg, Germany}{christian.glasser@uni-wuerzburg.de}{}{}
\authorrunning{D. Dingel and F. Egidy and C. Glaßer}
\Copyright{David Dingel and Fabian Egidy and Christian Glaßer}

\keywords{Computational Complexity, Promise Classes, Complete Sets, Oracle Construction}
\ccsdesc{Theory of computation~Problems, reductions and completeness}
\ccsdesc{Theory of computation~Oracles and decision trees}
\ccsdesc{Theory of computation~Complexity theory and logic}

\nolinenumbers

\begin{document}
\maketitle
\begin{abstract}
We construct an oracle relative to which $\NP = \PSPACE$,
but $\UP$ has no many-one complete sets.
This combines the properties of an oracle by Hartmanis and Hemachandra \cite{hh88}
and one by Ogiwara and Hemachandra \cite{oh93}.

The oracle provides new separations of classical conjectures
on optimal proof systems and complete sets in promise classes.
This answers several questions by Pudlák \cite{pud17}, e.g.,
the implications $\hUP \Longrightarrow \hCONN$ and
$\hSAT \Longrightarrow \hTFNP$ are false relative to our oracle.

Moreover, the oracle demonstrates that, in principle,
it is possible that $\TFNP$-complete problems exist,
while at the same time $\SAT$ has no p-optimal proof systems.
\end{abstract}

\section{Introduction}

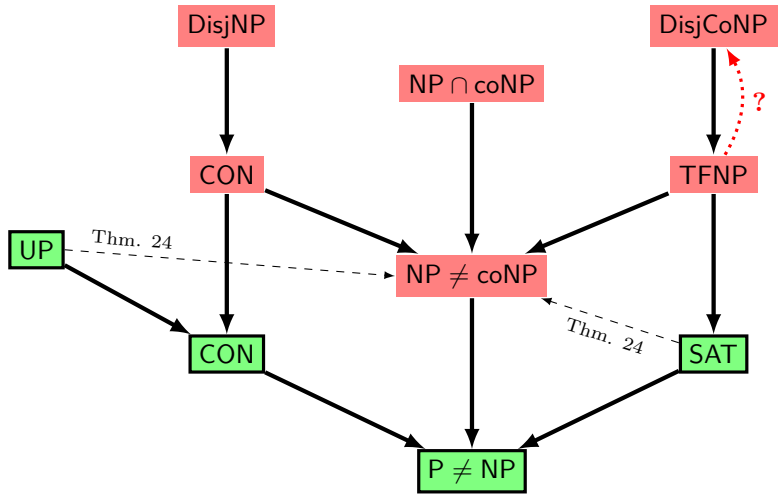
\begin{figure}[tb]
    \centering
    \begin{tikzpicture}[node/.style={anchor=base},tips=proper,
        EXPTRUE/.style={rectangle,draw=black,very thick,fill=green!50},
        EXPFALSE/.style={rectangle,very thick,fill=red!50},
        IMPTRUE/.style={rectangle,draw=black,very thick,fill=green!50},
        IMPFALSE/.style={rectangle,very thick,fill=red!50},
        ]
        \node [align=center,IMPTRUE] (pneqnp) at (0,0) {$\hPneqNP$};
        \node [above=2.0cm of pneqnp,EXPFALSE] (npeconp) {$\hNPneqcoNP$};
        \node [above=2.0cm of npeconp,IMPFALSE] (npnconp) {$\NPcoNP$};
        \node [above left=1.0cm and 2.0cm of pneqnp,IMPTRUE] (con) {$\hCON$};
        \node [above right=1.0cm and 2.0cm of pneqnp,IMPTRUE] (sat) {$\hSAT$};
        \node [above=4.0cm of con,anchor=base,IMPFALSE] (disjnp) {$\mathsf{DisjNP}$};
        \node [above=4.0cm of sat,anchor=base,IMPFALSE] (disjconp) {$\mathsf{DisjCoNP}$};
        \node [above=2.0cm of sat,anchor=base,IMPFALSE] (tfnp) {$\hTFNP$};
        \node [above=2.0cm of con,anchor=base,IMPFALSE] (conn) {$\hCON$};

        \node [above left=1.0cm and 2.0cm of con,anchor=base,EXPTRUE] (up) {$\mathsf{UP}$};
        \node [above right=1.0cm and 2.0cm of sat,anchor=base] (centering-dummy) {};

        \draw [-latex,ultra thick] (con) -- (pneqnp);
        \draw [-latex,ultra thick] (sat) -- (pneqnp);
        \draw [-latex,ultra thick] (npeconp) -- (pneqnp);
        \draw [-latex,ultra thick] (npnconp) -- (npeconp);
        \draw [-latex,ultra thick] (disjnp) -- (conn);
        \draw [-latex,ultra thick] (conn) -- (con);
        \draw [-latex,ultra thick] (up) -- (con);
        \draw [-latex,ultra thick] (disjconp) -- (tfnp);
        \draw [-latex,ultra thick] (tfnp)-- (sat);
        \draw [-latex,ultra thick] (tfnp)-- (npeconp);
        \draw [-latex,ultra thick] (conn)-- (npeconp);

        \draw [dotted, -latex, very thick, bend right, red ] (tfnp) edge  node [pos=0.5,right] {\bf{?}} (disjconp);

        \draw [dashed, -latex                      ] (up.east) edge  node [sloped,pos=0.2,above] {\scriptsize Thm. \ref{thm:results}} (npeconp.west);
        \draw [dashed, -latex                      ] (sat) edge  node [sloped,pos=0.5,below] {\scriptsize Thm. \ref{thm:results}} (npeconp);
    \end{tikzpicture}
    \caption{%
        Hypotheses that are true relative to our oracle are filled green and have borders, whereas those that are false relative to the oracle are red without borders.
        Solid arrows mean relativizable implications.
        A dashed arrow from one conjecture $\mathsf A$ to another conjecture $\mathsf B$ means that there is an oracle $X$ against the implication $\mathsf A\Rightarrow\mathsf B$, meaning that $\mathsf A \land \neg\mathsf B$ holds relative to $X$.
        For clarity, we only depict the two strongest separations proved in this paper, and omit those that either follow from these, as well as those already known before.
        All possible separations have now been achieved, except for the one depicted as the red dotted arrow.
    }
    \label{fig:diagram}
\end{figure}

We investigate relationships between the following classical conjectures
on $\NP$, optimal proof systems, and complete sets in promise classes.
All these conjectures have a long history and remain open.
They are deeply rooted in formal logic,
but also have far reaching practical implications.

\begin{center}
\begin{tabular}{>{\normalsize}r>{\normalsize}l}
    \hPneqNP: & \mbox{$\P$ does not equal $\NP$ \cite{coo71,lev73}}\\
    \hNPneqcoNP: & \mbox{$\NP$ does not equal $\coNP$ \cite{edm66}}\\[1ex]
    \hCON: & \mbox{p-optimal proof systems for $\TAUT$ do not exist \ \cite{kp89}}\\
    \hCONN\!: & \mbox{optimal proof systems for $\TAUT$ do not exist \ \cite{kp89}}\\ 
    \hSAT: & \mbox{p-optimal proof systems for $\SAT$ do not exist \ \cite{kp89}}\\[1ex]
    \hTFNP: & \mbox{$\TFNP$ does not contain many-one complete problems \ \cite{mp91}}\\
    \hNPcoNP: & \mbox{$\NP \cap \coNP$ does not contain many-one complete problems \ \nocite{sip82}\cite{kan79}}\\
    \hUP: & \mbox{$\UP$ does not contain many-one complete problems \ \cite{hh88}}\\
    \hDisjNP: & \mbox{$\DisjNP$ does not contain many-one complete pairs \ \cite{raz94}}\\
    \hDisjCoNP: & \mbox{$\DisjCoNP$ does not contain many-one complete pairs \ \cite{mes00,pud14}} \\
\end{tabular}
\end{center}
\noindent
We refer to Pudlák \cite{pud13} for a comprehensive elaboration of the conjectures and their context.

\subparagraph{Pudláks program.}
The known implications between these conjectures are shown in \autoref{fig:diagram}.
They raise the question of whether further implications
can be recognized with the currently available methods.
Pudlák \cite{pud17} therefore initiated a research program
to either prove further implications or to disprove them relative to an oracle.
The latter shows that the implication is beyond the reach of current methods.
Hence, from todays perspective,
the corresponding conjectures are not mere reformulations of one another,
but instead deserve individual attention.

So far such oracles have been constructed by
Verbitskii~\cite{ver91}, Glaßer et al.~\cite{gssz04}, Khaniki~\cite{kha22},
Dose~\cite{dos20c,dos20a,dos20b}, Dose and Glaßer~\cite{dg20},
and Egidy, Ehrmanntraut and Glaßer~\cite{eeg22}.

\subparagraph{Optimal proof systems.}
The hypotheses $\hCON$, $\hCONN$, and $\hSAT$ are concered
with the existence of (p-)optimal proof systems.
The study of proof systems was initiated by Cook and Reckhow \cite{cr79} and
motivated by the $\NP \overset{?}{=} \coNP$ question,
because the set of tautologies $\TAUT$ has a proof system with polynomially bounded proofs
if and only if $\NP = \coNP$.
Krajíček and Pudlák \cite{kp89,kra95,pud98} link questions about proof systems
to questions in bounded arithmetic. 
Especially the notions of optimal and p-optimal proof systems
have gained much attention in research.
They are closely connected to the separation of fundamental complexity classes, e.g.,
Köbler, Messner, and Torán \cite{kmt03} show that the absence of optimal proof systems for $\TAUT$
implies $\NEE \neq \coNEE$.
Moreover, (p-)optimal proof systems have tight connections to promise classes, e.g.,
if $\TAUT$ has p-optimal proof systems, then $\UP$ has many-one complete sets \cite{kmt03}.
Many more relationships are known between proof systems and promise classes
\cite{bey04, bey06, bey07, bey10, bkm09, bs11, gsz07, gsz09, kmt03, pud17, raz94}.

\subparagraph{TFNP-complete problems.}
The class $\TFNP$ contains a multitude of problems that are of central importance
to various practical applications, 
and as such is currently being researched with great intensity.
The search for a complete problem in $\TFNP$ is of primary concern to certain fields
such as cryptography, but it remains open whether such a problem exists.

For instance, the security of \rsa\ \cite{rsa} is primarily based
on the hardness of integer factorization,
which is a $\TFNP$ problem.
\rsa\ is now widely considered to no longer represent the state of the art in modern cryptography.
Hence researchers develop cryptoschemes whose security is based on the hardness of other problems.
However, it is not clear whether these are strictly harder than factorization.
A $\TFNP$-complete problem would allow cryptoschemes
that are at least as hard to break as any $\TFNP$ problem, i.e.,
schemes with optimal security guarantees.

Since the search for $\TFNP$-complete problems has been unsuccessful,
various subclasses were defined and studied in order to construct
at least a complete problem for those subclasses. 
Their definitions are based on the proof of totality used to show the membership to $\TFNP$.
Prominent examples are
$\pls$ \cite{jpy88}, based on the argument that \emph{every dag has a sink},
$\ppa$ \cite{pap91}, based on the fact that \emph{every graph has an even number of nodes with odd degree},
$\ppad$ \cite{pap91} and $\ppads$, based on slight variations of $\ppa$ for directed graphs,
and 
\ppp \cite{pap94}, based on the polynomial pigeonhole principle%
\footnote{%
    For precise definitions we refer to recent papers \cite{gp17}, \cite[notably Figure 1]{goos22}.%
}\!. 
All of them have complete problems \cite{jpy88,pap94} and can be defined in a syntactical way too \cite{pap94}.
Beame et al.\ \cite{bceip95} constructed oracles relative to which these classes
are no subset of $\P$, and no inclusions exist beyond the few that are already known. 
The classes are significant to various practical applications like
polynomially hard one way functions \cite{bpr15,gar16},
collision resistant hash functions \cite{gp17},
computing square roots modulo $n$ \cite{jer16}, and
the security of prominent homomorphic encryption schemes \cite{jkkz21}.

\subsection*{Our Contribution}

\subparagraph{1. Oracle with $\NP = \PSPACE$, but $\UP$ has no many-one complete sets}\hfill \\
Hartmanis and Hemachandra \cite{hh88} construct an oracle relative to which
$\UP$ does not have many-one complete sets.
Ogiwara and Hemachandra \cite{oh93} construct an oracle relative to which $\UP \neq \NP = \PSPACE$.
Our oracle provides both properties at the same time,
but this is achieved using completely different methods.
Due to the far-reaching collapse $\NP = \PSPACE$ and
the close connection between $\UP$-completeness and optimality of proof systems,
we obtain a number of useful properties summarized in Corollary \ref{cor:results}.

\subparagraph{2. Consequences for Pudlaks program}\hfill \\
Regarding the conjectures shown in \autoref{fig:diagram},
all known implications are presented in the figure,
while for some of the remaining implications
there exist oracles relative to which the implications do not hold.
From the implications that were left open, our oracle refutes all but one.
The key to this observation was to take the conjecture $\NP \neq \coNP$ into consideration.
As we make sure that our oracle satisfies $\NP = \coNP$,
the conjectures $\hCONN$ and $\hTFNP$ are false, 
while the conjectures $\hCON$ and $\hSAT$ are equivalent.
In addition, the oracle satisfies the conjecture $\hUP$,
which implies $\hCON$, and thereby in summary refutes all of the following previously open implications:
\begin{alignat*}{3}
     \hUP  &\Longrightarrow \hCONN     \qquad    \hUP  &&\Longrightarrow {\NP \neq \coNP}  \qquad  \hUP  &&\Longrightarrow \hDisjNP \\
     \hCON    &\Longrightarrow \hCONN     \qquad    \hCON    &&\Longrightarrow {\NP \neq \coNP}  \qquad  \hSAT &&\Longrightarrow \hDisjCoNP \\
     \hSAT &\Longrightarrow \hTFNP  \qquad    \hSAT &&\Longrightarrow {\NP \neq \coNP}  \qquad  \qquad  && 
\end{alignat*}
Therefore, all implications that are not presented in \autoref{fig:diagram}
fail relative to some oracle, except for one single implication,
which we leave as open question:
\[\hTFNP \mathop{\Longrightarrow}\limits^{?} \hDisjCoNP\]
We suspect that the construction of an oracle with $\hTFNP$ and $\neg \hDisjCoNP$
is a particular challenge.

\subparagraph{3. Consequences for TFNP-complete problems}\hfill \\
Regarding the search for $\TFNP$-complete problems,
our oracle demonstrates that, in principle, it is possible that
$\NP=\coNP$ and hence $\TFNP$-complete problems exist,
while at the same time $\SAT$ has no p-optimal proof systems.
To understand this situation, let us compare the following properties.
\begin{enumerate}
    \item $\NP=\coNP$\\
    In this case there exists a polynomial-time algorithm $A$
    such that for every $\NP$-machine $N$ the follwing holds:
    $N' := A(N)$ is a total $\NP$-machine and for all $x \in L(N)$,
    $N(x)$ and $N'(x)$ have the same set of accepting paths.
    This is a kind of parsimonious modification
    of faulty machines to globally enforce the promise.

    \item $\SAT$ has a p-optimal proof system\\
    In this case there exists
    an algorithm $B$ such that the following holds:
    For every total $\NP$-machine there exists a polynomial $p$ such that for all $x$,
    $B(N,x)$ finds within $p(|x|)$ steps a proof for the statement ``$N(x)$ accepts''.
    Thus each total $\NP$-machine $N$ admits polynomial-time constructive proofs for
    ``$N(x)$ accepts'' (i.e., the promise is locally kept).
\end{enumerate}
Relative to our oracle, the first property holds, but not the second one.
Hence algorithm $A$ exists, which for every total $\NP$-machine $N$
provides a total $\NP$-machine $N' := A(N)$ such that $N(x)$ and $N'(x)$
have the same set of accepting paths.
However, since the second property does not hold,
it is not clear how to obtain polynomial-time constructive proofs for
the statement ``$N'(x)$ accepts'' (because in general, we do not know the function that reduces $N'(x)$ to the $\TFNP$-complete problem).
In short, although $N'$ is total, we might not have a constructive proof for this.

\section{Preliminaries}\label{sec:prelim}

\subsection{Basic Definitions and Notations}
This section covers the notation of this paper and introduces well-known complexity theoretic notions.
\subparagraph{Words and sets.} 
All sets and machines in this paper are defined over the alphabet $\Sigma = \{ 0,1 \}$. 
$\Sigma^\ast$ denotes the set of all strings of finite length, and for a string $w \in \Sigma^*$ let $|w|$ denote the length of $w$. 
We write $\N$ for the set of non-negative integers, $\N^+$ for the set of positive integers, and $\emptyset$ for the empty set.
We write $\N[X]$ for the set of polynomials with natural coefficients.
The powerset of some set $A$ is denoted by $\mathcal{P}(A)$,
and its cardinality by $\card{A}$.
For a set $A$ and $k \in \N$ we write $A^{=k} \coloneqq \{x \in A \mid |x| =k\}$ as the subset of $A$ containing only words of length $k$.
We define this analogously for $\leq$.
For a clearer notation, we use the abbreviations $\Sigma^k \coloneqq {\Sigma^*}^{=k}$ and $\Sigma^{\leq k} \coloneqq {\Sigma^*}^{\leq k}$.
Let $<_{\text{lex}}$ denote the quasi-lexicographic (i.e., `shortlex') ordering of words over $\Sigma^*$, uniquely defined by requiring $0 <_{\text{lex}} 1$ and $v <_{\text{lex}} w$ if $|v| < |w|$. 
The operators $\cup$, $\cap$, and $\setminus$ respectively denote the union, intersection and set-difference.

\subparagraph{Functions.} 
Under the definition of $<_{\text{lex}}$ there is a unique order-isomorphism between $(\Sigma^*, <_{\text{lex}})$ and $(\N, <)$, which induces a polynomial-time computable, polynomial-time invertible bijection between $\Sigma^*$ and $\N$ denoted by $\enc$ (resp., $\enc^{-1}$ for the inverse).
Note that this definition is a variant of the dyadic representation.
The domain and range of a function $f$ are denoted by $\dom(f)$ and $\ran(f)$. 
For a function $f$ and $A \subseteq \dom(f)$, we define $f(A) \coloneqq \{f(x) \mid x \in A\}$. 
We define certain polynomial functions $p_i \colon \N \to \N$ for $i \in \N$ by $p_i(n) \coloneqq n^i + i$.

Let $\langle \cdot \rangle \colon \bigcup _{i \geq 0} (\Sigma^\ast)^i \to \Sigma^\ast$ be an injective, polynomial time computable and polynomial time invertible pairing function\footnote{Notice that the pairing function explicitly is not surjective,
    so invertibility is meant in the sense that $\ran(\langle \cdot \rangle) \in \P$, and there is a function $f \in \FP$ such that $\langle f(y)_1, \ldots, f(y)_n \rangle = y$ for all $y \in \ran(\langle \cdot \rangle)$.} such that 
$|\langle u_1, \dots , u_n \rangle | = 2(|u_1| + \cdots + |u_n| + n)+1$.
The pairing function has the following property.
\begin{claim}\label{claim:listencoding}
Let $w,x \in \Sigma^*$.
If $|w| \geq |x|+3$, then $|w| \geq |\langle x,w \rangle|/4$.
\end{claim}
\begin{proof}
The following calculation proves the claimed statement:
\[|\langle x,w \rangle| = 2(|x| + |w| + 2) + 1 = 2|x| + 5 + 2|w| \leq 2|w| + 2|w| = 4 |w|\qedhere\]
\end{proof}

\subparagraph{Machines.}
We use the default model of a Turing machine in the deterministic as well as in the non-deterministic variant, abbreviated by DTM, resp., NTM.
The language decided by a Turing machine $M$ is denoted by $L(M)$.
We use Turing transducers to compute functions.
For a Turing transducer $F$ we write $F(x)=y$ when on input $x$ the transducer outputs $y$.
We sometimes refer to the function computed by $F$ as `{the function $F$}'.

\subparagraph{Complexity Classes.}
The classes $\P$, $\NP$, $\coNP$, $\FP$ and $\PSPACE$ denote the standard complexity classes.
Furthermore, we are interested in several promise classes.
Originally defined by Valiant \cite{val76}, $\UP$ denotes the set of languages that can be decided by some NTM that has at most one accepting path on every input.
A pair $(A,B)$ is a disjoint $\NP$-pair, if $A,B \in \NP$ and $A \cap B = \emptyset$. Selman \cite{sel88} and Grollmann and Selman \cite{gs88} defined the class $\DisjNP$ consisting of all disjoint $\NP$-pairs.
Similarly, Fenner et al.\ \cite{ffnr96,ffnr03} defined disjoint $\coNP$-pairs and the class $\DisjCoNP$. Megiddo and Papadimitrou \cite{mp91} define the class $\TFNP$ as
\[\TFNP \coloneqq \{(R,p) \in (\mathcal{P}(\Sigma^* \times \Sigma^*) \cap \P) \times \N[x] \mid \forall y \in \Sigma^*~\exists x \in \Sigma^{\leq p(|y|)}\colon (x,y) \in R\}\]

\subparagraph{Reductions and Complete Problems.}
The common polynomial-time many-one reducibility for sets $A,B \subseteq \Sigma^*$ is denoted by $\leqmp$, i.e., $A \leqmp B$ if there exists an $f \in \FP$ such that $x \in A \Leftrightarrow f(x) \in B$.
For disjoint pairs, we say that a pair $(A,B)$ is polynomial-time many-one reducible to $(C,D)$, denoted by $(A,B) \leqmpp (C,D)$, if there is a function $f \in \FP$ such that $f(A) \subseteq C$ and $f(B) \subseteq D$ \cite{raz94}.
For $\TFNP$-problems, we use the reducibility described by Johnson, Papadimitrou, and Yannakakis \cite{jpy88}.
Namely, for problems $(T,q), (R,p) \in \TFNP$, $(T,q)$ is polynomial-time many-one reducible to $(R,p)$, denoted by $(T,q) \leq_{\TFNP} (R,p)$, if there exist $f,g \in \FP$ such that
\[\forall y \in \Sigma^* \colon \forall x \in \Sigma^{\leq p(|y|)} \colon (x,f(y)) \in R \Longrightarrow (g(x,y),y) \in T\]
If $\leq$ is some notion of reducibility for some class $\mathcal{C}$, then we call a set $A$ $\leq$-complete for $\mathcal{C}$, if $A \in \mathcal{C}$ and for all $B \in \mathcal{C}$ the reduction $B \leq A$ holds.
The set of satisfiable boolean formulas will be denoted by $\SAT$, and the set of boolean tautologies will be denoted by $\TAUT$. They are $\leqmp$ complete for $\NP$ (resp. $\coNP$). 

\subparagraph{Proof systems.} We use proof systems for sets as defined by Cook and Reckhow \cite{cr79}, i.e., a function $f \in \FP$ is a proof system for $\ran(f)$.
Analogously to the notion of reduction for complexity classes we define the terms (p-)simulation for proof systems:
a proof system $f$ (p-)simulates a proof system $g$, denoted by $g \sim f$ (resp., $g \psim f$), if there exists a total function $\pi$ (resp., $\pi \in \FP$) and a polynomial $p$ such that $|\pi(x)| \leq p(|x|)$ and $f(\pi(x))=g(x)$ for all $x \in \Sigma^*$.
We call a proof system $f$ (p-)optimal for $\ran(f)$ if $g \sim f$ (resp., $g \psim f$) for all $g \in \FP$ with $\ran(g)=\ran(f)$.

\subparagraph{Oracle specific definitions and notations.} An oracle $B$ is a subset of $\Sigma^*$.
We relativize the concept of Turing machines and Turing transducers by giving them access to a write-only oracle tape as proposed by Simon \cite{sim77}\footnote{When considering time-bounded computation, this machine model reflects the usual relativization of Turing machines. For space-bounded computations, the oracle tape is also subject to the space-bound.}.
Furthermore, we relativize complexity classes, proof systems, reducibilities and 
(p-)simulation
by defining them over machines with oracle access, i.e., whenever a Turing machine or Turing transducer is part of a definition, we replace them by an oracle Turing machine or an oracle Turing transducer.
We indicate the access to some oracle $B$ in the superscript of the mentioned concepts, i.e., $\P^B$, $\NP^B$, $\FP^B$, $\dots$ for complexity classes, $M^B$ for a Turing machine or Turing transducer $M$, $\leq_{\mathrm{m}}^{p,B}$, $\leq_\mathrm{m}^{\mathrm{pp},B}$, $\leq_{\TFNP}^B$ for the reducibilities and $\leq_\mathrm{s} ^B$ and $\leq_\mathrm{s} ^{\mathrm{p},B}$ for (p-)simulation.
We sometimes omit the oracles in the superscripts, e.g., when sketching ideas in order to convey intuition, but never in the actual proof.

Let $\{F_i\}_{i \in \N}$ be a standard enumeration of polynomial time oracle Turing transducers, such that relative to any oracle $B$, $F_i^B$ has running time exactly $p_i$ and for any function $f \in \FP^B$, there is some $i$ such that $F_i^B$ computes $f$.
Since the functions computed by $\{F_i^B\}_{i \in \N}$ form exactly $\FP^B$, we call such machines $\FP^B$-machines.
Similarly, let $\{N_i\}_{i \in \N}$ be a standard enumeration of polynomial time non-deterministic oracle Turing machines, such that relative to any oracle $B$, $N_i^B$ has running time exactly $p_i$ on each path and for any set $L \in \NP^B$, there is some $i$ such that $L(N_i^B) = L$.
Since the sets decided by $\{N_i^B\}_{i \in \N}$ form exactly $\NP^B$, we call such machines $\NP^B$-machines.
If $p$ is an encoding of a computation path,
then $Q(p)$ denotes the set of oracle queries on $p$.

\subparagraph*{Relativized QBF.} 
For any oracle $B$, we define the relativized problem $\QBF^B$ as the typical $\QBF$ problem as defined by Shamir \cite{sha90},
but the boolean formulas are extended by expressions $B(w)$ for $w \in \Sigma^*$, which evaluate to true if and only if $w \in B$.
It holds that for any oracle $B$, $\QBF^B$ is $\leq_\mathrm{m}^{\mathrm{p},B}$-complete for $\PSPACE^B$.

Throughout the remaining sections, $M$ will be a polynomial-space oracle Turing machine which, given access to oracle $B$, accepts $\PSPACE^B$.
Choose $k \geq 3$ such that $M$ on input $x \in \Sigma^\ast$ completes using at most $t(x) \coloneqq |x|^k + k$ space (including words written to the oracle tape).

\subsection{Notions for the Oracle Construction}
In this section we define all necessary objects and properties for the oracle construction and convey their intuition.
\subparagraph{Tools for \textsf{UP}.}
In order to achieve that $\UP^B$ has no complete set, 
we will ensure that for any set $L(N_i^B) \in \UP^B$ a set $W_i^B \in \UP^B$ exists such that $W_i^B$ does not reduce to $L(N_i^B)$, i.e., $W_i^B \not \leq_\mathrm{m}^{\mathrm{p},B} L(N_i^B)$.
Here, $W_i^B$ is a witness that $L(N_i^B)$ is not complete for $\UP^B$.
For this, we injectively assign countably infinite many levels of the oracle to each set $W_i$, where a level constitutes certain words of same length as follows.
\begin{definition}[$H_i$]\hfill\\
    Let $e(0) \coloneqq 2,
    ~e(i+1) \coloneqq 2^{e(i)}$ 
    and $H_i \coloneqq \{e(2^i \cdot 3^j) \mid j \in \N\}$ for $i \in \N$.
\end{definition}
\begin{corollary}[Properties of $H_i$]\hfill
\begin{enumerate}[(i)]
\item For every $i \in \N$, the set $H_i$ is countably infinite and a subset of the even numbers.\label{cor:Hi-i}
\item $H_i \in \P$ for all $i \in \N$.\label{cor:Hi-ii}
\item For $i,j \in \N$ with $i \neq j$, it holds that $H_i \cap H_j = \emptyset$.\label{cor:Hi-iii}
\end{enumerate}
\end{corollary}
\begin{definition}[Witness language $W_i^B$]\hfill\\
\label{definition_h_x}
    For $i \in \N$ and an oracle $B$, we define the set
    \[W_i^B \coloneqq \{ 0^n \mid n \in H_i \mbox{ and there exists } x \in \Sigma^n \mbox{ with } x \in B\}\]
\end{definition}
\begin{corollary}[Sufficient condition for $W_i^B \in \UP^B$]\hfill\\
    \label{cor:hg_eigenschaften}
    For any oracle $B$ and any $i \in \N$, the following implication holds:
    $$\card{B^{=n}} \leq 1~\mbox{for all } n \in H_i \Longrightarrow W_i^B \in \UP^B$$
\end{corollary}
\begin{proof}
    Let $B$ and $i$ be arbitrary.
Since $H_i \in \P$, an $\NP$-machine can reject on all inputs except $0^n$ with $n \in H_i$.
On inputs $0^n$ the machine can non-deterministically query all words $x \in \Sigma^n$ and accept if $x \in B$.
Now, if $\card{B^{=n}} \leq 1$, then this machine has at most one accepting path on all inputs.
Hence, $W_i^B \in \UP^B$. 
\end{proof}

We can control the membership of $W_i$ to $\UP$ in the oracle construction.
We will never add more than one word on the levels $H_i$, except for when we can rule out $N_i$ as an $\UP$-machine.

\subparagraph{Tools for $\boldsymbol{\NP = \PSPACE}$.}%
In order to achieve that $\NP^\B = \PSPACE^\B$, we will encode $\QBF^\B$ into the oracle $\B$, such that $\QBF^\B \in \NP^\B$, from which $\NP^\B = \PSPACE^\B$ follows. 

\begin{definition}[Coding set $Z^B$]\hfill\\
    For any oracle $B$, we define the set
    \[Z^B \coloneqq  \{ x \in \Sigma^\ast \mid \exists w \in \Sigma^{t(x)} \colon \langle x, w \rangle \in B \}\]
\end{definition}

\begin{corollary}\label{cor:z_in_np}
    For any oracle $B$, it holds that $Z^B \in \NP^B$.
\end{corollary}
We will assemble the oracle $\B$ in such a way that $Z^\B = \QBF^\B$, 
i.e., $Z^\B$ will be $\PSPACE^\B$-hard, and thus $\NP^\B = \PSPACE^\B$.

\subparagraph{Goals of the construction.}
\begin{definition}[Desired properties of $\B$]\label{def:desiredproperties}\hfill\\
The later constructed oracle $\B$ should satisfy the following properties:
\begin{enumerate}[\textbf{P}1]
\item $\forall x \in \Sigma^\ast \colon ( x \in \QBF^\B \Longleftrightarrow x \in Z^\B )$. \namedlabel{np_conp_sufficient}{P1}
\smallskip
\\
(Meaning: $\QBF^\B \in \NP^\B$.)
\item $\forall i \in \N$, at least one of the following statements holds: \namedlabel{satcon_sufficient}{P2}
    \begin{enumerate}[(I)]
        \item $\exists x \in \Sigma^\ast \colon N_i^\B(x)$ accepts on more than one path.
        \smallskip
        \\
         (Meaning: $N_i^\B$ is not a $\UP^\B$-machine.) \namedlabel{f_i_kaputt}{P2.I}
        \item Both of the following statements hold:
        \begin{enumerate}[(a)]
        \item $\forall n \in H_i \colon$ $\card{\B^{=n}} \leq 1$\namedlabel{gi_nicht_simuliert-i}{P2.II.a}
        \smallskip
        \\
        (Meaning: $W_i^\B \in \UP^\B$.)
        \item $\forall j \in \N ~ \exists x \in \Sigma^\ast \colon N_i^\B(F_j^\B(x))$ accepts if and only if $x \notin W_i^\B$.\namedlabel{gi_nicht_simuliert-ii}{P2.II.b}
		\smallskip
		\\
		(Meaning: $W_i^\B$ does not reduce to $L(N_i^\B)$ via $F_j^\B$.)
        \end{enumerate} 
 \namedlabel{gi_nicht_simuliert}{P2.II}
    \end{enumerate}
\end{enumerate}
\end{definition}
The following lemma shows that these properties imply our desired structural properties of complexity classes.
\begin{lemma}\label{lemma:properties}
Let $\B \subseteq \Sigma^\ast$.
\begin{enumerate}[(i)]
\item If statement \emph{\ref{np_conp_sufficient}} is satisfied relative to $\B$, then $\NP^\B = \PSPACE^\B$.\namedlabel{lemma:properties-i}{(i)}
\item If statement \emph{\ref{satcon_sufficient}} is satisfied relative to $\B$, then there is no $\leq_\mathrm{m}^{\mathrm{p},\B}$-complete set for $\UP^\B$.\namedlabel{lemma:properties-ii}{(ii)}
\end{enumerate}
\end{lemma}
\begin{proof}
To \ref{lemma:properties-i}: Recall that $\QBF^{\B}$ is $\leq_\mathrm{m}^{\mathrm{p},\B}$-complete for $\PSPACE^\B$.
Since property \ref{np_conp_sufficient} holds relative to $\B$, $Z^\B = \QBF^\B$.
From \autoref{cor:z_in_np} we get $Z^\B \in \NP^\B$ and thus $\NP^\B = \PSPACE^\B$.

To \ref{lemma:properties-ii}: Assume there is a $\leq_\mathrm{m}^{\mathrm{p},\B}$-complete set $C$ for $\UP^\B$.
Let $N_i^\B$ be an $\UP^\B$-machine such that $L(N_i^\B) = C$.
Then property \ref{f_i_kaputt} cannot hold for $i$, because $N_i$ would not be an $\UP$-machine.
Hence, property \ref{gi_nicht_simuliert} holds for $i$.
By \ref{gi_nicht_simuliert-i} and \autoref{cor:hg_eigenschaften}, we get $W_i^\B \in \UP^\B$.
Since $L(N_i^\B)$ is complete, there is some $j \in \N$ such that for all $x \in \Sigma^*$ we have
\[x \in W_i^\B \Longleftrightarrow F_j^\B(x) \in L(N_i^\B).\]
This is a contradiction to the satisfaction of property \ref{gi_nicht_simuliert-ii}, because there has to be some $x$ where this equivalence does not hold.
Hence, the assumption has to be false and there is no $\leq_\mathrm{m}^{\mathrm{p},\B}$-complete set for $\UP^\B$.
\end{proof}
\subparagraph{Stages.}
We choose certain stages which are concerned with dealing with property \ref{gi_nicht_simuliert-ii} by a function $m$.
We will call the values of $m$ `{$\hUP$-stages}', since only on these we will contribute to enforcing property \ref{gi_nicht_simuliert-ii} and consequently $\hUP^\B$.
The stages need to be sufficiently distant from another.
Among other things in order for them to not influence each other, but also because of other technical caveats that arise later.
We will refer back to this section whenever the requirements are used.
For now, it suffices to know that the stages will be chosen `{sufficiently large}'.

We choose a function $m : \N \times \N \to \N$, $(i,j) \mapsto m(i,j)$ which meets the requirements:
\begin{enumerate}[M1]
    \item $m(i,j) \in H_i$.\namedlabel{satstufen_form}{M1}
    \smallskip
    \\
    (Meaning: $\hUP$-stages $m(i, \cdot)$ are important to the witness-language $W_i$.) 
    \item $2^{m(i,j)/4} > (p_i(p_j(m(i,j))))^2 + p_i(p_j(m(i,j))) + 1$.\namedlabel{expo_suchraum}{M2}
    \smallskip
    \\
    (Meaning: The number of words of length $m(i,j)$ is far bigger than the number of words the computation $N_i(F_j(0^{m(i,j)}))$ can query.)
    \item $m(i, j) > m(i', j')~ \Longrightarrow ~m(i,j) > p_{i'}(p_{j'}(m(i',j')))$. \namedlabel{m_monoton}{M3}
    \smallskip
    \\
    (Meaning: The stages are sufficiently distant from another such that for no $i',j' \in \N$, the computation $N_{i'}(F_{j'}(0^{m(i',j')}))$ can query any word of the $\hUP$-stages following $m(i',j')$.)
\end{enumerate}

The function $m$ does not have to be computable, it only has to be total.
Such a candidate certainly exists, since $\card{H_i} = \infty$ for any $i$, and thus $m(i,j)$ can simply be chosen large enough that the requirements \ref{expo_suchraum} and \ref{m_monoton} are met. 
Notice that since $H_i$ only contains even numbers for all $i \in \N$, so does $m$, due to \ref{satstufen_form}.
This leads to the following observation.
\begin{observation}\label{obs:even-odd}
For all $k \in \ran(m)$:~$\ran(\langle \cdot \rangle) \cap \Sigma^k = \emptyset$, because $\langle \cdot \rangle$ only maps to odd lengths.
    In particular, if $|x| = m(i,j)$ for some $i,j \in \N$, then $x \notin \ran(\langle \cdot \rangle)$. 
\end{observation}
From this, we can derive that encodings for $Z$, i.e., words of the form $\langle \cdot, \cdot \rangle$, never have the same length as $\hUP$-stages. This will be a helpful argument to rule out interactions between $\hUP$-stages and encodings for $Z$.

\section{Oracle Construction}
In this section we will construct an oracle $\B$ such that $\hUP^{\B}$ has no complete set and $\NP^{\B} = \PSPACE^{\B}$. First, we sketch the idea of the construction. Then we define $\B$, followed by the proof that $\B$ is well-defined. Finally, we prove that $\B$ satisfies the desired properties from \autoref{def:desiredproperties}. 
\subparagraph{Construction of \B.}

We construct the oracle $\B$ sequentially.
For each $x \in \Sigma^*$, we decide whether to add words to the oracle.
We give a brief sketch of the construction and its ideas:
\medskip

To \ref{np_conp_sufficient}: Whenever the input $x$ has the form $\langle x',0^{t(x')} \rangle$, we add $\langle x',w \rangle$ for some $w \in \Sigma^{t(x')}$ to the oracle if and only if $x' \in \QBF$.
Since $M(x')$ cannot query words of length $|\langle x',w \rangle|$, the membership of $x'$ to $\QBF$ does not change after adding $\langle x',w \rangle$.
From then on, we never add shorter words to $\B$, thus the encoding persists correct for the finished oracle.
\medskip

To \ref{satcon_sufficient}: On every $\hUP$-stage $n \coloneqq m(i,j)$,
we try to achieve property \ref{f_i_kaputt} or \ref{gi_nicht_simuliert}.
For this, we differentiate between three cases.
\begin{enumerate}
\item\label{sketch:1} If $N_i(F_j(0^n))$ accepts relative to the oracle constructed so far, we leave the stage $n$ empty.
In subsequent iterations, we ensure that added words will not interfere with this accepting path.
This achieves property \ref{gi_nicht_simuliert}, because $0^n \notin W_i$ and $F_j(0^n) \in L(N_i)$.
\item\label{sketch:2} If $N_i(F_j(0^n))$ rejects relative to the oracle constructed so far, we add a word of length $n$ to the oracle such that this computation keeps rejecting.
This achieves property \ref{gi_nicht_simuliert}, because $0^n \in W_i$ and $F_j(0^n) \notin L(N_i)$.
Notice that finding appropriate words may not be possible.
\item\label{sketch:3} If $N_i(F_j(0^n))$ rejects relative to the oracle constructed so far, and 
there is no choice of a word to add such that $N_i(F_j(0^n))$ keeps rejecting (i.e., case \ref{sketch:2} is not possible), 
then we force $N_i(F_j(0^n))$ to accept on two different paths (which is possible, as we will show).
In subsequent iterations, we ensure that added words will not interfere with these accepting paths.
This achieves property \ref{f_i_kaputt}. 
\end{enumerate}
In iterations after a $\hUP$-stage, if we have to add words to maintain \ref{np_conp_sufficient}, we choose them in such a way that property \ref{satcon_sufficient} still remains upheld for the previous stage.
This is captured in \autoref{prod:z}, \ref{prod:z:name}.
Once the next stage is reached, our word length will have increased enough such that the stage before cannot be affected anymore.

\begin{definition}[Oracle construction]\hfill\\\label{def:b}
Define $\B_0 \coloneqq \emptyset$.
For $k \in \N^+$, define ${\B _{k+1} \coloneqq \B_k \cup \emph{\ref{prod:constr:name}}(\enc(k), \B_k)}$ according to \autoref{prod}.
Finally, define $\B \coloneqq \bigcup _{k \in \N} \B_k$.
\end{definition}
\begin{lstlisting}[caption={\namedlabel{prod:constr:name}{\constprod}\ref{prod:constr:name}\texttt{(}$\texttt{x,B}$\texttt{)}},label={prod},abovecaptionskip=-\medskipamount,mathescape=true,escapechar=?,numbers=left,backgroundcolor=\color{white}]
If $\tn{x} \tn{ = } \langle \tn{x',0}^{\tn{t(x')}} \rangle$ and $\tn{M}^{\tn{B}}\tn{(x')}$ accepts:?\label{line:constr:z:condition}?
    If the largest UP-stage $\tn{m(i,j)} \tn{ < } \tn{|x|}$ exists:?\label{line:constr:z:start}?
        Return ?\ref{prod:z:name}?(x',i,j,B)
    Return $\{\tn{x}\}$?\label{line:constr:z:end}? 
Elif $\tn{x = 0}^\tn{m(i,j)}$ for some $\tn{i,j} \in \N$:?\label{line:constr:up:condition}?
    Return ?\ref{prod:up:name}?(i,j,B)?\label{line:constr:up:return}?
Else:
  Return $\emptyset$. 
\end{lstlisting}

\begin{lstlisting}[caption={\namedlabel{prod:up:name}{\upprod}\ref{prod:up:name}\texttt{(i,j,B)}}, % handles \ref{satcon_sufficient}
                   label={prod:up},
                   abovecaptionskip=-\medskipamount, mathescape=true, escapechar=?, numbers=left,backgroundcolor=\color{white}]
n $\coloneqq$ m(i,j)
If $\tn{N}_\tn{i}^{\tn{B}}\tn{(}\tn{F}_\tn{j}^{\tn{B}}\tn{(}\tn{0}^\tn{n}\tn{))}$ accepts:
  Return $\emptyset$?\label{line:up:accepts}\Suppressnumber?
// Here, $\tn{N}_\tn{i}^{\tn{B}}\tn{(F}_\tn{j}^{\tn{B}}\tn{(0}^\tn{n}\tn{))}$ rejects?\Reactivatenumber?
If $\exists \tn{y} \in \Sigma^{\tn{n}}$ s.th. $\tn{N}_\tn{i}^{\tn{B} \cup \{\tn{y}\}}\tn{(}\tn{F}_\tn{j}^{\tn{B} \cup \{\tn{y}\}}\tn{(}\tn{0}^\tn{n}\tn{))}$ rejects:
  Return $\{\tn{y}\}$?\label{line:up:keepsrejecting}?
Else:
  Let $\tn{y,z} \in \Sigma^{\tn{n}}$ s.th. $\tn{N}_\tn{i}^{\tn{B} \cup \{\tn{y,z}\}}\tn{(F}_\tn{j}^{\tn{B} \cup \{\tn{y,z}\}}\tn{(0}^\tn{n}\tn{))}$ accepts on $\geq \tn{2}$ paths ?\label{line:up:twopaths:choose}?
  Return $\{\tn{y,z}\}$?\label{line:up:twopaths:return}?
\end{lstlisting}

\begin{lstlisting}[caption={\namedlabel{prod:z:name}{\zprod}\ref{prod:z:name}\texttt{(x',i,j,B)} % handles \ref{np_conp_sufficient}, while sustaining \ref{satcon_sufficient}
},
                   label={prod:z},
                   abovecaptionskip=-\medskipamount, mathescape=true, escapechar=?, numbers=left,backgroundcolor=\color{white}]
n $\coloneqq$ m(i,j)
If $\tn{N}_\tn{i}^{\tn{B}}\tn{(F}_\tn{j}^{\tn{B}}\tn{(0}^\tn{n}\tn{))}$ accepts on at least one path $\tn{p}_\tn{1}$:?\label{line:z:accepts:condition}\smallskip?
  $\tn{Q} \coloneqq \begin{cases}
      \tn{Q(}\tn{p}_\tn{1}\texttt{)} \cup \tn{Q(}\tn{p}_\tn{2}\texttt{)} 
      & \texttt{ if }\tn{N}_\tn{i}^{\tn{B}}\tn{(F}_\tn{j}^{\tn{B}}\tn{(0}^\tn{n}\tn{))}\texttt{ accepts on another path }\tn{p}_\tn{2} \neq \tn{p}_\tn{1}
      \\ 
      \tn{Q(}\tn{p}_\tn{1}\texttt{)}
      & \texttt{ else }
  \end{cases}$?\label{line:z:accepts:start}?
  Choose $\tn{w} \in \Sigma^{\tn{t(x')}}$ s.th. $\langle \tn{x',w} \rangle \notin \tn{Q}$?\label{line:z:accepts:choose}?
  Return $\{\langle \tn{x',w} \rangle\}$?\label{line:z:accepts:return}\smallskip\Suppressnumber?
// Here, $\tn{N}_\tn{i}^{\tn{B}}\tn{(F}_\tn{j}^{\tn{B}}\tn{(0}^\tn{n}\tn{))}$ rejects?\Reactivatenumber?
If $\exists \tn{y} \in \{\langle \tn{x',w} \rangle \mid \tn{w} \in \Sigma^{\tn{t(x')}}\}$ s.th. $\tn{N}_\tn{i}^{\tn{B} \cup \{\tn{y}\}}\tn{(F}_\tn{j}^{\tn{B} \cup \{\tn{y}\}}\tn{(0}^\tn{n}\tn{))}$ rejects:
  Return $\{\tn{y}\}$?\label{line:z:keepsrejecting}?
Let $\tn{y,z} \in \{\langle \tn{x',w} \rangle \mid \tn{w} \in \Sigma^{\tn{t(x')}}\}$ s.th. $\tn{N}_\tn{i}^{\tn{B} \cup \{\tn{y,z}\}}\tn{(F}_\tn{j}^{\tn{B} \cup \{\tn{y,z}\}}\tn{(0}^\tn{n}\tn{))}$ accepts on $\geq \tn{2}$ paths?\label{line:z:twopaths:choose}?
Return $\{\tn{y,z}\}$?\label{line:z:twopaths:return}?        
\end{lstlisting}

Notice that in \ref{prod:up:name} reaching line \ref{line:up:accepts} corresponds to case \ref{sketch:1} of the construction sketch for \ref{satcon_sufficient},
reaching line \ref{line:up:keepsrejecting} corresponds to case \ref{sketch:2}, 
and reaching lines \ref{line:up:twopaths:choose}--\ref{line:up:twopaths:return} corresponds to case \ref{sketch:3}.
In lines \ref{line:constr:z:start}--\ref{line:constr:z:end} of \ref{prod:constr:name}, because $M$ accepts, we need to add at least one word $\langle x',w \rangle$ to the oracle in order to achieve \ref{np_conp_sufficient}.
\autoref{prod:z}, \ref{prod:z:name}, ensures that a previous UP-stage either remains unaffected by this (lines \ref{line:z:accepts:condition}--\ref{line:z:keepsrejecting}), or at least $i$ now satisfies property \ref{f_i_kaputt} instead (lines \ref{line:z:twopaths:choose}--\ref{line:z:twopaths:return}).

It is not clear that the oracle construction can be performed as stated, 
because lines \ref{line:z:accepts:choose} and \ref{line:z:twopaths:choose} in \ref{prod:z:name}
and line \ref{line:up:twopaths:choose} in \ref{prod:up:name} 
claim the existence of words with complex properties.
The following claims show that these lines can be performed as stated for the calls to \ref{prod:constr:name} in the oracle construction.
\begin{claim}[Line \ref{line:up:twopaths:choose} in \ref{prod:up:name} can be performed]\label{claim:n-two-accept-paths-possible}\hfill\\
    If the line \ref{line:up:twopaths:choose} is reached in \ref{prod:up:name} in the step \ref{prod:constr:name}($\enc(k),\B_k$), then $y$ and $z$ can be chosen as stated.
\end{claim}
\begin{proof}
    Let $x \coloneqq \enc(k)$. When reaching this line, then for some $i,j \in \N$~with ${n \coloneqq m(i,j)}$ we have that 
    $N_i^{\B_k}(F_j^{\B_k}(0^n))$ rejects, 
    and for all extensions $\B_k \cup \{x'\}$ with $x' \in \Sigma^{n}$, 
    $N_i^{\B_k \cup \{x'\}}(F_j^{\B_k \cup \{x'\}}(0^n))$ accepts.
    Since the accepting paths appear after adding $x'$ to $\B_k$, these paths must query $x'$.
    There are $2^n$ words of length $n$ and an accepting path can query at most $p_i(p_j(n))$ of these words.
    Let $X$ be a set consisting of $(p_i(p_j(n)))^2+p_i(p_j(n)) + 1$ pairwise different words $x'$ of length $n$.
    By \ref{expo_suchraum}, 
    \[2^{n/4} > (p_i(p_j(n)))^2+p_i(p_j(n)) + 1,\]
    whereby such a set exists.
    We show that there must be two different words $y, z \in X$ 
    such that $N_i^{\B_k \cup \{y\}}(F_j^{\B_k \cup \{y\}}(0^n))$ (resp., $N_i^{\B_k \cup \{z\}}(F_j^{\B_k \cup \{z\}}(0^n))$) accepts 
    and does not query $z$ (resp., $y$) on some accepting path. 

    To choose $y$, we look at $X' \subseteq X$ consisting of $p_i(p_j(n))+1$ pairwise different words.
    When extending $\B_k$ by a word from $X'$ for each word in $X'$ separately, the resulting leftmost accepting paths query at most 
    \[p_i(p_j(n)) \cdot \card{X'} = p_i(p_j(n)) \cdot (p_i(p_j(n)) + 1) = (p_i(p_j(n)))^2 + p_i(p_j(n))\]
    different words in total.
    Hence, there is at least one unqueried word $y \in X$. 

    To choose $z$, consider $N_i^{\B_k \cup \{y\}}(F_j^{\B_k \cup \{y\}}(0^n))$, which queries at most $p_j(p_i(n))$ words on the leftmost accepting path.
    Hence, there is some unqueried word $z \in X$. 

    So, $N_i^{\B_k \cup \{y\}}(F_j^{\B_k \cup \{y\}}(0^n))$ accepts on a path which queries $y$ but not $z$, and $N_i^{\B_k \cup \{z\}}(F_j^{\B_k \cup \{z\}}(0^n))$ accepts on a path which queries $z$ but not $y$.
    Thus, $y \not = z$, both accepting paths are different and both are preserved when extending by both $y$ and $z$, resulting in at least two different accepting paths for $N_i^{\B_k \cup \{y, z\}}(F_j^{\B_k \cup \{y,z\}}(0^n))$.
\end{proof}
\begin{claim}[Line \ref{line:z:twopaths:choose} in \ref{prod:z:name} can be performed]\label{claim:c-two-accept-paths-possible}\hfill\\
If the line \ref{line:z:twopaths:choose} is reached in \ref{prod:z:name} in the step \ref{prod:constr:name}($\enc(k),\B_k$), then $y$ and $z$ can be chosen as stated.
\end{claim}
\begin{proof}
Let $x \coloneqq \enc(k)$. When reaching this line, 
then for $x = \langle x',0^{t(x')} \rangle$ 
and the biggest $\hUP$-stage $n \coloneqq m(i,j) \leq |x|$ 
we have that $N_i^{\B_k}(F_j^{\B_k}(0^n))$ rejects,
but for all extensions $\B_k \cup \{y\}$ with $y \in \{\langle x', w \rangle \mid w \in \Sigma^{t(x')}\} \eqqcolon Y$, 
the computation $N_i^{\B_k \cup \{y\}}(F_j^{\B_k \cup \{y\}}(0^n))$ accepts.
Since the accepting paths appear after adding $y$ to $\B_k$, these paths must query $y$.
By Claim \ref{claim:listencoding},  $|t(x')| \geq |x|/4 \geq n/4$.
Hence, there are $\geq 2^{n/4}$ words $w \in \Sigma^{t(x')}$, and thus $\card{Y} \geq 2^{n/4}$.
An accepting path can query at most $p_i(p_j(n))$ of these words.
Let $X$ be a set consisting of $(p_i(p_j(n)))^2+p_i(p_j(n)) + 1$ pairwise different words $y \in Y$.
By \ref{expo_suchraum}, 
\[2^{n/4} > (p_i(p_j(n)))^2+p_i(p_j(n)) + 1,\]
whereby such a set exists. From here on, we can proceed exactly as in the proof of Claim \ref{claim:n-two-accept-paths-possible} to find two fitting words $y$ and $z$.\qedhere
\end{proof}
\begin{claim}[Line \ref{line:z:accepts:choose} in \ref{prod:z:name} can be performed]\label{claim:fix-accept-path}\hfill\\
If the line \ref{line:z:accepts:choose} is reached in \ref{prod:z:name} in the step \ref{prod:constr:name}($\enc(k),\B_k$), then $w$ can be chosen accordingly.
\end{claim}
\begin{proof}
Let $x \coloneqq \enc(k)$. When reaching this line, then for $x = \langle x', 0^{t(x')} \rangle$ and the biggest $\hUP$-stage ${n \coloneqq m(i,j) \leq |x|}$ it holds that $N_i^{\B_k}(F_j^{\B_k}(0^n))$ accepts. In the worst case, there are at least two accepting paths.
Up to two accepting paths $p_1$ and $p_2$ query at most $2p_i(p_j(n))$ words in total.
Together with \ref{expo_suchraum}, we get
\[\card{Q(p_1) \cup Q(p_2)} \leq 2p_i(p_j(n)) \leq (p_i(p_j(n)))^2 + p_i(p_j(n)) < 2^{n/4}.\]
By Claim \ref{claim:listencoding}, $|t(x')| \geq |x|/4 \geq n/4$ and thus
\[\card{\Sigma^{t(x')}} \geq 2^{n/4}  > (p_i(p_j(n)))^2 + p_i(p_j(n)) + 1 \geq \card{Q(p_1) \cup Q(p_2)}.\]
Consequently, there is some $w \in \Sigma^{t(x')}$ with $\langle x',w \rangle \notin (Q(p_1) \cup Q(p_2))$. 
\end{proof}
The Claims \ref{claim:n-two-accept-paths-possible}, \ref{claim:c-two-accept-paths-possible}, and \ref{claim:fix-accept-path} show that $\B _{k+1} \coloneqq \B_k \cup \texttt{construct}(\enc(k),\B_k)$ is well-defined.
We make three further observations to the procedure \texttt{construct}.
\begin{observation}\label{obs:monotoneconstruction}
    For $k \in \N$, \emph{\ref{prod:constr:name}}$(\enc(k),\B_k)$ only adds words of length $|\enc(k)|$.
\end{observation}
The next two observations follow from \autoref{obs:monotoneconstruction} combined with \autoref{obs:even-odd}.
\begin{observation}\label{obs:0n-onetime}
Let $n \in H_i$ for some $i \in \N$ and $k \coloneqq \enc^{-1}(0^n)$.
Only the step \texttt{\emph{construct}}$(0^n,\B_k)$ can add words of length $n$ to $\B$.
\end{observation}
\begin{observation}\label{obs:codword-onetime}
Let $x \coloneqq \langle x', 0^{t(x')} \rangle$ for $x' \in \Sigma^*$ and $k \coloneqq \enc^{-1}(x)$.
Only the step \texttt{\emph{construct}}$(x,\B_k)$ can add words of the form $\langle x',w \rangle$ with $w \in \Sigma^{t(x')}$ to $\B$.
\end{observation}
\subparagraph{Proving the Properties.}
Finally, we show that the properties \ref{np_conp_sufficient} and \ref{satcon_sufficient} hold relative to $\B$.
We start with property \ref{np_conp_sufficient}.
\begin{lemma}\label{lem:prob1}
The oracle $\B$ satisfies property \emph{\ref{np_conp_sufficient}}.
\end{lemma}
\begin{proof}
    First, recall that by definition of $Z$, 
    \begin{equation}\label{eq:def-z}
        x \in Z^\B \Longleftrightarrow \exists w \in \Sigma^{t(x)}\colon \langle x, w \rangle \in \B 
    \end{equation}
    Let $x \in \Sigma^*$ be arbitrary,
    and consider the step of the oracle construction where $\B_k$ is defined as $\ref{prod:constr:name}(\langle x, 0^{t(x)} \rangle, \B_{k-1})$.
    By Observation \ref{obs:codword-onetime}, only \ref{prod:constr:name}$(\langle x, 0^{t(x)} \rangle, \B_{k-1})$ may add words of the form $\langle x, w \rangle$ for $w \in \Sigma^{t(x)}$ to $\B$.
    From this we can draw $x \notin Z^{\B_{k-1}}$, and further, together with $\B_k \subseteq \B$, 
    \begin{equation}\label{eq:z_only_at_k}
        \exists w \in \Sigma^{t(x)}: \langle x, w \rangle \in \B \Longleftrightarrow \exists w \in \Sigma^{t(x)}\colon \langle x, w \rangle \in \B_k 
    \end{equation}
    The lines \ref{line:constr:up:condition} and \ref{line:constr:up:return} in \ref{prod:constr:name}$(\langle x, 0^{t(x)} \rangle, \B_{k-1})$ are skipped, 
    since, by \ref{satstufen_form}, $m(i,j)$ is always even, whereas $\langle x,0^{t(x)}\rangle$ has odd length (see \autoref{obs:even-odd}). 
    Hence, \ref{prod:up:name} will not be entered. 
    Since at the lines \ref{line:constr:z:start}--\ref{line:constr:z:end} of \ref{prod:constr:name} always at least one word is added to the oracle, 
    \ref{prod:constr:name}$(\langle x, 0^{t(x)} \rangle, \B_{k-1})$ adds a word $\langle x, w \rangle$ with $w \in \Sigma^{t(x)}$ to $\B_{k-1}$ if and only if the condition in line \ref{line:constr:z:condition} is met and these lines are entered, i.e., if $M^{\B_{k-1}}(x)$ accepts.
    This gives
    \begin{equation}\label{eq:2}
    x \in \QBF^{\B_{k-1}} \Longleftrightarrow \exists w \in \Sigma^{t(x)}\colon \langle x, w \rangle \in \B_k
    \end{equation} 
    Since $M^{\B_{k-1}}(x)$ has a space bound of $t(x)$, it cannot ask words longer than $t(x)$. 
    Thereby $M^{\B_{k-1}}(x)$ cannot recognize a transition to the oracle $\B_k$, because we only add words $\langle x, w \rangle$ with $w \in \Sigma^{t(x)}$, which have length $> t(x)$.
    So, $M^{\B_k}(x)$ accepts if and only if $M^{\B_{k-1}}(x)$ accepts. 
    For the same reason, $M^{\B_{k}}(x)$ cannot recognize a transition to the oracle $\B$, because
    by Observation \ref{obs:monotoneconstruction}, all words in the set $\B \setminus \B_k$
    have length $\geq |\langle x, 0^{t(x)} \rangle|$, and thus $M^\B(x)$ accepts if and only if $M^{\B_k}(x)$ accepts.
    Combining these two facts gives
    \begin{equation}\label{eq:sat-indifferent}
        x \in \QBF^\B \Longleftrightarrow x \in \QBF^{\B_{k-1}} 
    \end{equation}
    In total, we conclude
    \begin{align*}
        x \in \QBF^\B 
        \overset{\text{\eqref{eq:sat-indifferent}}}{\Longleftrightarrow} x \in \QBF^{\B_{k-1}} 
        &\overset{\text{\eqref{eq:2}}}{\Longleftrightarrow} \exists w \in \Sigma^{t(x)}\colon \langle x, w \rangle \in \B_k \\
        &\overset{\text{\eqref{eq:z_only_at_k}}}{\Longleftrightarrow} \exists w \in \Sigma^{t(x)}\colon \langle x, w \rangle \in \B \\
        &\overset{\text{\eqref{eq:def-z}}}{\Longleftrightarrow} x \in Z^\B 
    \end{align*}
    Since $x \in \Sigma^*$ was chosen arbitrarily, property \ref{np_conp_sufficient} holds.
\end{proof}
Before proving \autoref{lem:prob2}, i.e., proving that property \ref{satcon_sufficient} holds, we prove the helpful claim that \texttt{construct} does not alter certain accepting paths.
\begin{claim}\label{claim:keep-accept}
    Let $i,j,k \in \N$ such that $k > \enc^{-1}(0^{m(i,j)})$, and $n \coloneqq m(i,j)$. 
    If $N_i^{\B_k}(F_j^{\B_k}(0^n))$ accepts (resp., accepts on more than one path), then so does $N_i^{\B_{k+1}}(F_j^{\B_{k+1}}(0^n))$.
\end{claim}
\begin{proof}
    If $\enc(k) \geq_{\text{lex}} 0^{p_i(p_j(n))+1}$, then by Observation \ref{obs:monotoneconstruction}, only words of length $> p_i(p_j(n))$ are added to $\B_k$.
    Hence, all paths of $N_i^{\B_k}(F_j^{\B_k}(0^n))$ and $N_i^{\B_{k+1}}(F_j^{\B_{k+1}}(0^n))$ are the same, from which the claimed statement follows.

    Otherwise, $\enc(k) \leqlex 1^{p_i(p_j(n))}$.
    By \ref{m_monoton}, $m(i,j)$ has to be the biggest $\hUP$-stage $\leq |\enc(k)|$.
    Consider the step \ref{prod:constr:name}($\enc(k),\B_{k}$).
    Either $\B_{k+1} = \B_{k}$ and the claim holds, or some words are added to $\B_{k}$.
    The procedure \ref{prod:up:name} is not entered, because $\enc(k) >_{\text{lex}} 0^{n}$, so the condition in line \ref{line:constr:up:condition} of \ref{prod:constr:name} is false, 
    whereby any added words would need to be added by \ref{prod:z:name}.
    Further, since $N_i^{\B_{k}}(F_j^{\B_{k}}(0^{n}))$ accepts, this can only happen via the lines \ref{line:z:accepts:start} to \ref{line:z:accepts:return}.
    Here, line \ref{line:z:accepts:choose} makes sure that $N_i^{\B_{k+1}}(F_j^{\B_{k+1}}(0^{n}))$ remains accepting (resp., remains accepting on more than one path), since the respective paths do not query the chosen words.
\end{proof}
\begin{corollary}\label{cor:keep-accept}
    Let $i,j,k \in \N$ such that $k > \enc^{-1}(0^{m(i,j)})$,
    and $n \coloneqq m(i,j)$.
    If $N_i^{\B_k}(F_j^{\B_k}(0^n))$ accepts (resp., accepts on more than one path), then so does $N_i^{\B}(F_j^{\B}(0^n))$.
\end{corollary}

\begin{lemma}\label{lem:prob2}
The oracle $\B$ satisfies property \emph{\ref{satcon_sufficient}}.
\end{lemma}
\begin{proof}
Let $i \in \N$ be arbitrary.
Assume that for $i$ the property \ref{f_i_kaputt} does not hold.
Then we show that property \ref{gi_nicht_simuliert} holds. 
\begin{claim}\label{claim:p1}
The property \ref{gi_nicht_simuliert-i} holds.
\end{claim}
\begin{proof}
Let $n \in H_i$ be arbitrary and $k \coloneqq \enc^{-1}(0^n)$.
By Observation \ref{obs:0n-onetime}, only the step \ref{prod:constr:name}($0^n,\B_k$) adds words of length $n$ to $\B$.
Since $n$ has even and $\langle \cdot, \cdot \rangle$ has always odd length (see \autoref{obs:even-odd}), the condition in line \ref{line:constr:z:condition} of \ref{prod:constr:name} evaluates to `{false}'.
So, only if line \ref{line:up:twopaths:return} in \ref{prod:up:name} is executed, more than one word of length $n$ is added to $\B_k$ and consequently to $\B$.
In this case, since the sets $H_0, H_1, \dots$ are disjoint, there is some $j \in \N$ with $m(i,j)=n$ where $N_i^{\B_{k+1}}(F_j^{\B_{k+1}}(0^n))$ accepts on more than one path.
Invoking Corollary \ref{cor:keep-accept} with $i$, $j$ and $k+1$, we get that also $N_i^{\B}(F_j^{\B}(0^n))$ accepts on more than one path.
This is a contradiction to the assumption that the property \ref{f_i_kaputt} does not hold for $i$.
Hence, line \ref{line:up:twopaths:return} in \ref{prod:up:name} cannot be executed during \ref{prod:constr:name}($0^n,\B_k$).
This gives $\card{\B^{=n}} \leq 1$, as required.
\end{proof}
\begin{claim}\label{claim:p2}
The property \ref{gi_nicht_simuliert-ii} holds.
\end{claim}
\begin{proof}
Let $j \in \N$ be arbitrary, $n \coloneqq m(i,j)$, and $k \coloneqq \enc^{-1}(0^n)$.
Consider the step \ref{prod:constr:name}($0^n,\B_k$).
Since $n$ has even length, the condition in line \ref{line:constr:z:condition} of \ref{prod:constr:name} evaluates to `{false}'. But the condition in line \ref{line:constr:up:condition} is satisfied, so \ref{prod:up:name} will define which words are added to the oracle.
Also recall that only at this step words of length $n$ are added to $\B$ (Observation \ref{obs:0n-onetime}), i.e., $\B_{k+1}^{=n} = \B^{=n}$.
We distinguish between the three cases on how $\B_{k+1}$ can be defined in \ref{prod:up:name}.
\medskip

If $\B_{k+1}$ is defined via line \ref{line:up:accepts}, then $\B^{=n} = \emptyset$ and $N_i^{\B_{k+1}}(F_j^{\B_{k+1}}(0^n))$ accepts.
By Corollary \ref{cor:keep-accept}, $N_i^{\B}(F_j^{\B}(0^n))$ accepts too.
Hence, $N_i^\B(F_j^\B(0^n))$ accepts and $0^n \notin W_i^\B$, i.e., $0^n$ satisfies property \ref{gi_nicht_simuliert-ii}.
\medskip

If $\B_{k+1}$ is defined via line \ref{line:up:twopaths:return}, then $N_i^{\B_{k+1}}(F_j^{\B_{k+1}}(0^n))$ accepts on more than one path.
By Corollary \ref{cor:keep-accept}, $N_i^{\B}(F_j^{\B}(0^n))$ also accepts on more than one path.
But then, property \ref{f_i_kaputt} holds for $i$, a contradiction to the assumption at the start of Lemma \ref{lem:prob2}.
Hence, this case cannot occur.
\medskip

If $\B_{k+1}$ is defined by line \ref{line:up:keepsrejecting}, then $\card{\B^{=n}} = 1$ and $N_i^{\B_{k+1}}(F_j^{\B_{k+1}}(0^n))$ rejects.
Let $k' > k+1$ be the smallest number such that $N_i^{\B_{k'}}(F_j^{\B_{k'}}(0^n))$ accepts.
Either $k'$ does not exist and hence, $N_i^\B(F_j^\B(0^n))$ rejects, satisfying property \ref{gi_nicht_simuliert-ii}. 

Otherwise $\enc(k') \leqlex 1^{p_i(p_j(n))}$, because $N_i(F_j(0^n))$ can query only words of length $\leq p_i(p_j(n))$, 
and for bigger $k'$, only words of length $>p_i(p_j(n))$ are added (Observation \ref{obs:monotoneconstruction}).
Consider the step \ref{prod:constr:name}$(\enc(k'-1),\B_{k'-1})$.
By \ref{m_monoton}, $m(i,j)$ is the biggest $\hUP$-stage $\leq |\enc(k'-1)|$.
Since 
\[0^n = \enc(k) \lelex \enc(k+1) \leqlex \enc(k'-1),\] 
the condition in line \ref{line:constr:up:condition} is not met, thus \ref{prod:up:name} is skipped.
Since by the minimality of $k'$ the computation $N_i^{\B_{k'-1}}(F_j^{\B_{k'-1}}(0^n))$ rejects, $\B_{k'}$ can only be defined via line \ref{line:z:twopaths:return} in \ref{prod:z:name}.
But then, $N_i^{\B_{k'}}(F_j^{\B_{k'}}(0^n))$ accepts on more than one path.
Invoking Corollary \ref{cor:keep-accept} for $i$, $j$, $k'$, we get that also $N_i^{\B}(F_j^{\B}(0^n))$ accepts on more than one path.
Consequently, property \ref{f_i_kaputt} holds for $i$, a contradiction to the assumption at the start of Lemma \ref{lem:prob2}.
Hence, this case cannot occur.
\end{proof}
Since $i$ is arbitrary, the Claims \ref{claim:p1} and \ref{claim:p2} show that if property \ref{f_i_kaputt} does not hold, property \ref{gi_nicht_simuliert} does.
Thus, property \ref{satcon_sufficient} holds.
\end{proof}

\section{Conclusion}
The oracle from the previous section gives the following result.
\begin{theorem}\label{thm:results}
There exists an oracle relative to which $\UP$ has no $\leqmp$-complete sets and $\NP = \PSPACE$.
\end{theorem}
\begin{proof}
This follows from the Lemmas \ref{lemma:properties}, \ref{lem:prob1}, and \ref{lem:prob2}. 
\end{proof}

We summarize the most important properties of the oracle from \autoref{thm:results} relating to the hypotheses of Pudlák \cite{pud17} (\ref{cor:props-01}--\ref{cor:props-09}, \ref{cor:props-15}, and \ref{cor:props-23}), general structural properties (\ref{cor:props-10}--\ref{cor:props-14}) and function classes connected to $\TFNP$ (\ref{cor:props-16}--\ref{cor:props-22}).
\begin{corollary}\label{cor:results}
    Relative to the oracle from \autoref{thm:results}, all of the following hold:
    \begin{enumerate}
        \item\label{cor:props-01} $\NP = \coNP$, i.e., $\neg \hNPneqcoNP$.
        \item\label{cor:props-02} $\UP$ does not contain $\leqmp$-complete sets, i.e., $\hUP$.
        \item\label{cor:props-03} $\TAUT$ does not have p-optimal proof systems, i.e., $\hCON$.
        \item\label{cor:props-04} $\SAT$ does not have p-optimal proof systems, i.e., $\hSAT$.
        \item\label{cor:props-05} $\TFNP$ has a complete problem, i.e., $\neg \hTFNP$.
        \item\label{cor:props-06} $\TAUT$ has optimal proof systems, i.e., $\neg \hCONN$.
        \item\label{cor:props-07} $\DisjNP$ has $\leqmpp$-complete pairs, i.e., $\neg \hDisjNP$.
        \item\label{cor:props-08} $\DisjCoNP$ has $\leqmpp$-complete pairs, i.e., $\neg \hDisjCoNP$.
        \item\label{cor:props-09} $\NP \cap \coNP$ has $\leqmp$-complete sets, i.e., $\neg \NPcoNP$.
        \item\label{cor:props-10} $\P \subsetneq \UP \subsetneq \NP \cap \coNP = \NP$.
        \item\label{cor:props-11} $\NE$ and $\NEE$ are closed under complement (classes defined in \emph{\cite{kmt03}}).
        \item\label{cor:props-12} $\NEE \cap \TALLY \not \subseteq \EE$ (classes defined in \emph{\cite{kmt03}}).
        \item\label{cor:props-13} $\NP$ and $\coNP$ have the shrinking property \emph{\cite[Def.~1.1]{grs11}}.
        \item\label{cor:props-14} $\NP$ and $\coNP$ do not have the separation property \emph{\cite[Def.~1.1]{grs11}}.
        \item\label{cor:props-15} $\DisjNP$ and $\DisjCoNP$ contain $\mathrm{P}$-inseparable pairs (defined in \emph{\cite{ffnr03}}).
        \item\label{cor:props-23} The \emph{ESY} conjecture does not hold \emph{\cite{esy84}}, i.e., there is no disjoint $\NP$-pair $(A,B)$ that is $\leq_\mathrm{T}^{\mathrm{pp}}$-hard for $\NP$. 
        \item\label{cor:props-16} $\NPSV_t \not \subseteq \PF$ \emph{\cite[Def.~1]{ffnr03}}.
        \item\label{cor:props-17} $\NPbV_t \not \subseteq_c \PF$ \emph{\cite[Def.~1]{ffnr03}}.
        \item\label{cor:props-18} $\NPkV_t \not \subseteq_c \PF$ for all $k \geq 2$ \emph{\cite[Def.~1]{ffnr03}}.
        \item\label{cor:props-19} $\NPMV_t \not \subseteq_c \PF$ \emph{\cite[Def.~1]{ffnr03}}.
        \item\label{cor:props-20} $\TFNP \not \subseteq_c \PF$.
        \item\label{cor:props-21} The conjecture $Q$ does not hold \emph{\cite[Def.~2]{ffnr03}}.
        \item\label{cor:props-22} The conjecture $Q'$ does not hold \emph{\cite[Def.~3]{ffnr03}}.
    \end{enumerate}
\end{corollary}
\begin{proof}
    To \ref{cor:props-01} and \ref{cor:props-02}: Follows from \autoref{thm:results}.
    \\
    To \ref{cor:props-03}: Follows from \ref{cor:props-02} by Köbler, Messner, and Torán \cite[Cor.~4.1]{kmt03}.
    \\
    To \ref{cor:props-04}: Follows from \ref{cor:props-03} and \ref{cor:props-01}.
    \\
    To \ref{cor:props-05}: Follows from \ref{cor:props-01} by Megiddo and Papadimitrou \cite[Thm.~2.1]{mp91}. 
    \\
    To \ref{cor:props-06}: Follows from \ref{cor:props-01} and the fact that $\NP$ has optimal proof systems \cite[Thm.~3.1]{mes00}.
    \\
    To \ref{cor:props-07}: Follows from \ref{cor:props-06} by Köbler, Messner, Torán \cite[Cor.~6.1]{kmt03}.
    \\
    To \ref{cor:props-08}: Follows from \ref{cor:props-07} and \ref{cor:props-01}.
    \\
    To \ref{cor:props-09}: Follows from \ref{cor:props-01} and that $\NP^\B$ has $\leqmp$-complete sets.
    \\
    To \ref{cor:props-10}: Follows from $\P \subseteq \UP \subseteq \NP$, $\P$ has complete sets, \ref{cor:props-02}, $\NP$ has complete sets and \ref{cor:props-01}.
    \\
    To \ref{cor:props-11}: Via padding, \ref{cor:props-01} implies $\NE = \coNE$ and $\NEE = \coNEE$.
    \\
    To \ref{cor:props-12}: Follows from \ref{cor:props-03} by Köbler, Messner, and Torán \cite[Cor.~7.1]{kmt03}.
 	\\
    To \ref{cor:props-13}: Follows from \ref{cor:props-01} by Glaßer, Reitwießner, and Selivanov \cite[Thm.~3.4]{grs11}.
    \\
    To \ref{cor:props-14}: Follows from $\UP \not \subseteq \coNP$ (implied by \ref{cor:props-10}), \ref{cor:props-01}, and Glaßer, Reitwießner, and Selivanov \cite[Thm.~3.9]{grs11}. 
    \\
    To \ref{cor:props-15}: Follows from $\P \neq \NP \cap \coNP$, implied by \ref{cor:props-10}. Otherwise we could $\mathrm{P}$-separate $(L,\overline{L})$ for all $L \in \NP \cap \coNP$, thus, $L \in \P$. 
    \\
    To \ref{cor:props-23}: Follows from \ref{cor:props-01}.
    \\
    To \ref{cor:props-16}: Follows from $\P \neq \NP \cap \coNP$ (implied by \ref{cor:props-10}), which is equivalent to \ref{cor:props-16}, as shown by Fenner et al.~\cite[Prop.~1]{ffnr03}.
    \\
    To \ref{cor:props-17}: Follows from \ref{cor:props-15} and the equivalence to the existence of $\mathrm{P}$-inseparable $\DisjCoNP$-pairs, as shown by Fenner et al.~\cite[Thm.~4]{ffnr03}. 
    \\
    To \ref{cor:props-18}: By \ref{cor:props-17}, $\NPbV_t \not \subseteq_c \PF$.
    Fenner et al.~\cite[Thm.~14]{ffnr03} show that this implies $\NPkV_t \not \subseteq_c \PF$ for all $k \geq 2$.
    \\
    To \ref{cor:props-19}: Follows from $\P \neq \NP \cap \coNP$ (implied by \ref{cor:props-10}) as shown by Fenner et al.~\cite[Thm.~2]{ffnr03}.
    \\
    To \ref{cor:props-20}: Follows from \ref{cor:props-19}, as shown by Fenner et al.~\cite[Prop.~7, Thm.~2]{ffnr03}.
    \\
    To \ref{cor:props-21}: Follows from \ref{cor:props-19} by Fenner et al.~\cite[Thm.~2]{ffnr03}.
    \\
    To \ref{cor:props-22}: Follows from \ref{cor:props-17} by Fenner et al.~\cite[Thm.~4]{ffnr03}.
\end{proof}

\bibliography{9-Literatursammlung}

\end{document}